\documentclass[11pt,reqno]{amsart}

\numberwithin{equation}{section}

\usepackage{amsmath, amsfonts, amssymb, amscd}

\usepackage{amssymb, amscd,enumerate,
color}

\usepackage{amsmath, amsfonts, amssymb, amscd, enumerate}
\usepackage{prev-macros}
\usepackage{mymacros}
\usepackage{paralist}
\usepackage{graphicx}

\usepackage{bbm}

\newcommand{\mytitlebib}[1]{\textit{#1}}

\theoremstyle{remark}
\newtheorem{hypothesis}{\normalfont\itshape Hypothesis}

\title[Finite reductions and inertial manifolds]
{Approximate Inertial manifold approach to non-equilibrium thermodynamics}
\author[Cardin, Favretti, Lovison]
{  Franco Cardin \qquad Marco Favretti \qquad Alberto Lovison \\
\address{\hspace{-0.42cm}Dipartimento di Matematica ``Tullio Levi-Civita'',
Universit\`a degli Studi di Padova,
Via Trieste 63, 35121 Padova,
ITALY}}
\begin{document}

\begin{abstract}In this paper a reaction-diffusion type equation is the starting point for setting up a genuine thermodynamic reduction, i.e. involving a finite number of parameters  or collective variables, of the initial system. This program is carried over by firstly operating a finite Lyapunov-Schmidt reduction of the cited reaction-diffusion equation when reformulated as a variational problem. In this way we gain an approximate finite-dimensional o.d.e. description of the initial system which preserves the gradient structure of the original one and that is similar to the approximate  inertial manifold description of a p.d.e. introduced by Temam and coworkers.
Secondly, we resort to  the stochastic version of the o.d.e., taking into account in this way the uncertainty (loss of information)  introduced with the above mentioned reduction. We study this reduced stochastic system using classical tools from large deviations, viscosity solutions and weak KAM Hamilton-Jacobi theory.
In the last part we highlight  some essential similarities existing between our approach and the comprehensive treatment  non equilibrium thermodynamics given  by Jona-Lasinio and coworkers. The starting point of their axiomatic theory --motivated by large deviations description of lattice gas models-- of systems in a stationary non equilibrium state is precisely a conservation/balance law which is akin  to our simple model of reaction-diffusion equation.
\end{abstract}

\maketitle

\tableofcontents

\centerline{\today}

\setcounter{section}{0}
\setcounter{subsection}{1}
\section{Introduction}
Classical equilibrium thermodynamics, in its essence, aims at describing the equilibrium states of a system composed of many particles  using a small number of variables. The individuation of the right  number of macroscopic variables, often called  \emph{collective variables} \cite{Boldrighini:1987ez,Maragliano:2006mf,Muller2002393,Yukhnovskiui:1987ft}, is a non trivial task. 

This  problem has received attention in the chemical physics community: if the system is described by a Fokker-Planck equation,  the existence of a \emph{large gap} in the spectrum of the associated elliptic operator   is interpreted, from a physical point of view, as the emergence of a finite dimensional description of the system in terms of its leading eigenvalues.  This  is, roughly speaking, the content of the Kramers-Klein theory,  which dates back to the seminal paper \cite{K}; see e.g.\,\cite{M} and \cite{MC}  for an exposition and some applications of this framework.
Nowadays a sound justification of the above Kramers-Klein theory relies on the Witten theory \cite{Witten},  also discussed  in Arnold-Keshin \cite{AK}.

Another reductionistic approach which is relevant for the present work  is the   Amann--Conley--Zehnder
reduction (ACZ in what follows) 
\cite{Amann:1980kr,
Amann:1979mq,
Conley:1976bf}. 
ACZ is a global Lyapunov-Schmidt type reduction transforming an infinite dimensional variational principle  into an equivalent finite dimensional one. 
Namely, a static PDE variational problem becomes perfectly equivalent to a finite algebraic system, see Section \ref{sec:ACZ}.
This method has been employed in conjunction with topological techniques
for proving results of existence and multiplicity of solutions for nonlinear differential equations
\cite{Cardin:2008ix,Cardin:2014,
Viterbo:1990wj}.
Notably for our purposes, we have that the ACZ philosophy can be applied also to dissipative dynamical equations, giving rise to the notion of \emph{inertial manifolds} \cite{T}.
An inertial manifold is a finite dimensional smooth structure which contains the attractor of the system and which is also invariant with respect to the evolution equations. Moreover, every orbit of the system starting outside of the manifold is exponentially rapidly attracted towards the manifold itself. This feature supports the idea that the global motion could be effectively studied and controlled by studying the motion  restricted to the inertial manifold.  
In principle, the application of this methodology is extremely appealing and promising for a wide range of applications. However the existence of such manifolds is not guaranteed in the general case and furthermore the numerical exploitation of such structures is usually problematic. 
Nevertheless, given the different scales at which the modeled phenomena occur, some of the fine details of the motion can be neglected without affecting the general and long time behavior of the system.  It is possible therefore to define an asymptotic series expansion of such manifolds, and such expansion have a simple explicit formulation and can be implemented numerically in an effective and efficient  way.  Truncations in this series expansions are called \emph{approximate inertial manifolds} (AIM) and have been applied successfully in a range of situations and for a wide spectrum of phenomena  \cite{T2,T3,T}. 
Eventually,  AIMs can be defined and employed for studying the global motion even when an exact inertial manifold does not exist. 

Coming back to thermodynamics, an even more challenging task is  the description of a system in an out of equilibrium state. In this realm the classical thermodynamic description, except for the case of moderate deviation form equilibrium, is less powerful and even the status of basic notions as non equilibrium free energy or entropy is debated.  

In a series of papers  \cite{MFT,JL2009,PRL,IL}  G. Jona-Lasinio and coworkers develop a theory aiming at a macroscopic description of thermodynamic systems in a stationary non equilibrium state (driven diffusive systems) in terms of a finite number of space--time  fields representing pertinent thermodynamic variables. The theory has an axiomatic format that stems from generalization of previous rigorous result for lattice gas systems where the macroscopic equations are the hydrodynamic limit equations. A thorough treatment of large deviations for these systems was given by Varadhan and others \cite{KL:1999,Varadhan:1984}.  For a driven diffusive system the macroscopic equations are balance equations and for these Jona-Lasinio and co-workers develop a theory of macroscopic fluctuations which is a generalization to fields variables of the Friedlin and Wentzell theory  of large deviations for ODE. One of the most interesting features of the theory is that the infinite dimensional analog $\mathcal{F}$ -see $(\ref{enlibera})$- of the \emph{quasi potential} of the Friedlin and Wentzell theory is interpreted as the non equilibrium free energy of the system in a stationary non equilibrium state, without any assumption of moderate deviations from equilibrium. 

The guiding principle of this paper is the following: starting from an infinite dimensional system (PDE) akin to those considered by Jona-Lasinio we operate a reduction and we study a stochastic version of the resulting finite dimensional ODE. We can thus make use of the  standard Friedlin and Wentzell Theory. 
Here is the sketch of our strategy.
  
\begin{itemize}
\item We start from a reaction-diffusion PDE which is also a $L^2-$\emph{gradient type} system, see (\ref{RD}), thought as a simplified prototype model for the balance law governing the evolution of a open system in a stationary non equilibrium state.

\item We  operate a finite dimensional reduction of it obtaining in this way a finite dimensional ODE, equation $(\ref{gra})$,  still of \emph{gradient type}. The true dynamics approaches the reduced one and the  difference between them can be estimated by the associated approximate inertial manifold, see sect.\,\ref{AIM1}, \ref{AIM2}.

\item We mimic the loss of information linked to the finite approximate reduction  above  introduced adding  a stochastic noise to this deterministic ODE $(\ref{gra})$ and  we consider the associated probabilistic Fokker-Planck equation. A careful description of the limiting behavior of the SDE is afforded by   the standard Friedlin-Wentzell theory of large deviations for stochastic ODE.  A major role in  this theory is played by the \emph{action functional} over the dynamic evolution of the system and its associated variational principle which defines the so-called \emph{quasi-potential} -see $(\ref{V})$ and $(\ref{veq})$- The related Hamilton-Jacobi equation is investigated in this paper with viscosity and weak KAM theory techniques. 

 \end{itemize}

We are confident that the results of this approach could shed light on the more elaborate infinite dimensional analogs  of the Jona-Lasinio theory. 

For example, in the finite dimensional case,  using asymptotic theory of viscosity solutions we are able to prove that the right hand member of the finite-dimensional H-J equation $(\ref{SHJ})$ is   the  Ma\~n\'e critical value,  which  is zero in the present case, as we resumed in Prop.\,\ref{Prop1}. The need of setting equal to zero the right hand member of $(\ref{infinite})$  is not obvious in the Jona-Lasinio formulation.

Next,  for the gradient type ODE $(\ref{gra})$ we compute explicitly the quasi potential $V^{eq}_\infty$ and we show that it coincides with potential energy $W$ whose gradient $-\nabla W$ gives the gradient reduced dynamics. 

Moreover, in sect.\,\ref{confronto}  we  propose a  comparison  between the infinite dimensional free energy $\mathcal{F}$ and a finite dimensional analog of it  -see $(\ref{enlibera2})$- which is expressed in terms of the quasi-potential  $V^{eq}_\infty$.

As a final remark, we want to draw attention on the fact that our finite-dimensional gradient equation describes the evolution of projection over a finite dimensional subspace of and evolutive PDE, that gives rise to a finite dimensional ODE. The abstract nature of this ODE is on the one hand   the major limit to extend the comparisons between the initial physical problem and its reduced dynamic and on the other hand the source of its effectiveness to give a ``thermodynamic type'' description of the system in terms of ``a few'' collective variables.

\setcounter{equation}{0}

\vskip 1 cm

\par
\section{Gradient dynamics  from exact  finite reduction} 
Let $\Delta$ be the elliptic Laplace-Beltrami operator on a closed Riemann manifold $(\Omega, g)$, $|v|^2=\langle v,v\rangle_g$ and  $V:\R\to \R$ be the potential energy function. Our aim  is to reconsider the reaction-diffusion equation 
\begin{equation}\label{RD}
 \frac{\partial u}{\partial t}=\Delta u-V'(u),
\end{equation}
thought as a simplified prototype model for the balance law governing the evolution of the field in the macroscopic theory of driven diffusive systems in non stationary thermodynamics,  see \cite{JL2009}, equation $(\ref{axiom3})$ below. Even though (\ref{RD}) does not admit a standard variational  formulation, the search of its equilibria does. Indeed, considering the following real valued functional\footnote{E.g.  $\cH=W^{2,2}_0(\Omega;\bR)$} $\mathcal{J}: \cH \rightarrow \bR$, 
\begin{equation}\label{PF}
\mathcal{J}(u)=\int_{\Omega}\left[\frac{1}{2}|\nabla u(x)|^2+V(u(x))\right]dx
\end{equation}
The variational principle
\begin{equation}\label{ELI}
{\mathcal{J}}'(u)h=-\int_{\Omega}\left(\Delta u-V'(u)\right)h=0\quad \forall h\in \cH
\end{equation}
produces the
Euler-Lagrange equation  
\begin{equation}\label{EL}
\Delta u=V'(u)
\end{equation}

Moreover, again by using $\mathcal{J}$, equation $(\ref{RD})$ can  also be read in a distributional $L^2$ weak format, as an infinite dimensional gradient system
\begin{equation}\label{RD-w}
\int_{\Omega}\frac{\partial u}{\partial t}h=\int_{\Omega}\left(\Delta u-V'(u)\right)h=-\mathcal{J}'(u)h,\quad \forall h\in C^\infty_0,
\end{equation}
whose solutions would be `running' to the equilibria solutions of (\ref{EL}). The equation (\ref{RD-w}) appears as  a \textit{gradient descent}  
equation and its equilibrium solutions  are the critical points of (\ref{PF}).

Notice that  the `static' functional $\mathcal{J}$ in  (\ref{PF}) plays the natural role of candidate Lyapunov functional:

\begin{equation}\label{ILF}
\frac{d}{dt}\mathcal{J} (u)={\mathcal{J}}'(u) \frac{\partial u}{\partial t}=-\int_{\Omega}\left(\Delta u-V'(u)\right) \frac{\partial u}{\partial t}=
-\int_{\Omega}\left|\Delta u-V'(u)\right|^2\leq 0
\end{equation}

\subsection{A global Lyapunov-Schmidt reduction }\label{sec:ACZ}
Suppose that the $\Delta$-spectral representation of $\cH$ does work,
\begin{equation}
\Delta u_j=-\lambda_j u_j, \quad u_j|_{\partial \Omega}=0,
\quad\langle u_i , u_j\rangle= \int_{\Omega}u_i u_j=\delta_{ij},
\quad \lambda_0=0<\lambda_1\leq  \dots 
\end{equation}
and then consider an `a priori' cut-off $m\in \bN$,
\begin{equation}
\cH={\bP}_m \cH\oplus {\bQ}_m \cH\ni u=\mu+\eta,
\end{equation}
\begin{equation}
\mu={\bP}_m u=\sum_{j=1}^m \langle   u, u_j   \rangle u_j, \qquad \eta={\bQ}_m u=\sum_{j>m} \langle   u, u_j   \rangle u_j
\end{equation}
Denoting by $g(f)$ the unique solution to the Poisson equation (the so-called $\Delta^{-1}$)
 \begin{equation}\label{g}
g(f)=-\sum_{j>0}\frac{\langle   f, u_j   \rangle}{\lambda_j   }u_j, \qquad \Delta g(f)=f, \quad f\in \cH,
\end{equation}
we translate (\ref{EL}) into a fixed point problem:
\begin{equation}\label{EL1}
u=g(V'(u))
\end{equation}
and consider the following decomposition of (\ref{EL1}),
\begin{equation}        \label{dec}
\left\{
\begin{array}{rcl}
\mu &=&\bP_m g (V'(    \mu +\eta))   \\  \\
\eta &= &\bQ_m g (V'(    \mu +\eta))   
\end{array}
\right.
\end{equation}
Look at (\ref{dec}$)_2$: we will check that if $V'$ is \textit{globally Lipschitz},
\begin{equation}
C=\text{Lip}(V')<+\infty,
\end{equation}
then we can choose the cut-off $m$ such that the following map, from $\bQ_m\cH$ into itself, for any fixed $\mu\in\bP_m \cH$, is a contraction:
$$
\bQ_m\cH \longrightarrow \bQ_m\cH
$$
\begin{equation}
\eta\longmapsto \bQ_m g (V'(    \mu +\eta))   
\end{equation}
To see this, we have that for any $\eta_1, \eta_2\in \bQ_m\cH$, recalling (\ref{g}),
$$
\Vert
\bQ_m g (V'(    \mu +\eta_1))   -\bQ_m g (V'(    \mu +\eta_2)) \Vert\leq \frac{C}{\lambda_m}\Vert \eta_1 -\eta_2     \Vert
$$
and since the sequence $\lambda_0=0<\lambda_1\leq \dots $ is growing and   unbounded, we can choose $m$ so that
 \begin{equation}
\frac{C}{\lambda_m}<1
\end{equation}
Denote by $\tilde \eta (\mu)$ the unique fixed point function such that, for any $\mu\in \bP_m\cH$,
\begin{equation}\label{fp}
\tilde \eta (\mu)=\bQ_m g (V'(    \mu +\tilde \eta (\mu)))
\end{equation}
Finally, we see that  the solutions of the finite part (\ref{dec}$)_1$ are related to a finite  variational setting.
Indeed, defined the real valued function
\begin{equation}\label{Fr}
W:\bR^m \to \bR, \qquad W(\mu):=\mathcal{J}( \mu +\tilde \eta (\mu)),
\end{equation}
we see that
$$
dW(\mu)d\mu={\mathcal{J}}'(\mu +\tilde \eta (\mu))(d\mu+d\tilde \eta (\mu)d\mu)=
$$
$$
=-\langle   \Delta u-V'(u)\Big|_{u=\mu +\tilde \eta (\mu)}  ,   d\mu+d\tilde \eta (\mu)d\mu                  \rangle= $$

$$
=-\langle   \Delta u-V'(u)\Big|_{u=\mu +\tilde \eta (\mu)}  ,   \bP_m d\mu+\bQ_m d\tilde \eta (\mu)d\mu                  \rangle= 
$$
$$
=-\langle   \Delta u-V'(u)\Big|_{u=\mu +\tilde \eta (\mu)}  ,   \bP_m d\mu\rangle=-\langle   \Delta u-V'(u)\Big|_{u=\mu +\tilde \eta (\mu)}  ,   d\mu\rangle,
$$
\medskip

\begin{equation}\label{Fin}
dW(\mu)d\mu=-\langle   \Delta u-V'(u)\Big|_{u=\mu +\tilde \eta (\mu)}  ,   d\mu\rangle
\end{equation}
In other words, the critical points $\mu$ of $W$ are \textit{exactly} the solutions of (\ref{dec}$)_1$. Moreover, we can say that all the solutions of the variational principle $\mathcal{J}'=0$, are obtained by the above finite reduction ${\nabla W}=0$: insert a solution $\mu$ of (\ref{dec}$)_1$ in (\ref{fp}), obtaining (\ref{dec}$)_2$, then add $\mu$ to the resulting $\tilde \eta (\mu)$, recovering $u$, solution of $\mathcal{J}'=0$.

\subsection{A finite gradient equation}
The  proposal of finite  reduction of the dynamic equation (\ref{RD}), recalling (\ref{ELI}), (\ref{Fr}) and (\ref{Fin}), goes as follows.

Let us consider the above $\bP_m$ and $\bQ_m$ projections of  (\ref{RD}) --or (\ref{RD-w})--
\begin{equation}        \label{dec1}
\left\{
\begin{array}{rcl}
\mu_t &=&\Delta \mu -\bP_m  V'(    \mu +\eta)   \\  \\
\eta_t &= &\Delta \eta -\bQ_m  V'(    \mu +\eta) 
\end{array}
\right. 
\end{equation}
As an approximation of it,  we compose any occurrence of $\eta$ in (\ref{dec1}$)_1$ by the above fixed point function $\tilde\eta (\mu)$ in (\ref{fp}), so that it becomes 
$$
\mu_t =\Delta \mu -\bP_m  V'(    \mu +\tilde\eta (\mu) )
$$
and necessarily  (\ref{dec1}$)_2$ has to be substituted consistently by 
$$
0=\Delta \tilde\eta (\mu) -\bQ_m  V'(    \mu +\tilde\eta (\mu))
$$
which is identically satisfied in view of (\ref{fp}).
Eventually, we have obtained 
\begin{equation}        \label{dec2}
\left\{
\begin{array}{rcl}
\mu_t &=&\Delta \mu -\bP_m  V'(    \mu +\tilde\eta (\mu) )\\  \\
0&= &\Delta \tilde\eta (\mu) -\bQ_m  V'(    \mu +\tilde\eta (\mu))
\end{array}
\right.
\end{equation}
We see that (\ref{dec2}$)_1$, in view of (\ref{Fin}), is precisely 
\begin{equation}
\label{gra}
\frac{d\mu}{d t}(t)=-{\nabla W}(\mu(t))
\end{equation}
We point out that for the  above outlined reduced equation the set of critical points is  in a one-to-one correspondence with the equilibria of the original reaction--diffusion equation. If ${\mu}^*$ is a critical point, minimum for $W$, then $W$ is a (local) Lyapunov function for the ODE dynamics (\ref{gra}) around ${\mu}^*$:
$$
\dot W={\nabla W}\cdot \dot\mu=-|{\nabla W}|^2\leq 0
$$

A main feature of this approximate reduction is that it preserves the 
gradient-like character of  the PDE (\ref{RD}). 
We point out that it can be interpreted as a sort of ``quasi-static'' thermodynamic version of (\ref{RD}), precisely for the presence of ``zero'' in the l.h.s. of (\ref{dec2}$)_2$. 
\subsection{Inertial Manifolds (IM) and Approximate Inertial Manifold (AIM)}\label{AIM1}
It turns out that the finite exact reduction scheme introduced above has strong connections with the theory of inertial and approximate inertial manifolds developed by R. Temam and coworkers \cite{RT,T,T2}. In order to compare the effectiveness of both theories to deal with our problem,
we briefly review the fundamental facts regarding inertial manifolds and their approximated counterparts and  then we conclude showing how to  use in our setting an  important estimate developed in  the AIM framework.

\smallskip

Consider an arbitrary finite or infinite dynamical system 
\begin{equation}\label{eq:dynsys1}
	\deinde{u}{t} = \mathcal{F}(u) \qquad \tonde{e.g., \qquad \mathcal{F}(u) = \Delta u - V'(u), \qquad \text{as in (\ref{RD})}},
\end{equation}
defined on a Hilbert space $\cH$, associated to a semigroup $\set{S(t)}_{t\geq 0}$, where $S(t)$ is the mapping 
\begin{equation}
	S(t): u_0 \longmapsto u(t) = S(t) u_0, 
\end{equation}
and $u(t)$ is the solution of (\ref{eq:dynsys1}) such that $u(0) = u_0$. 

\begin{definition}
 An \emph{inertial manifold} $\mathcal{M}\subseteq \cH$ of the system (\ref{eq:dynsys1}) is a finite dimensional Lipschitz manifold with the following properties:
\begin{enumerate}[$(i)$]
	\item $\mathcal{M}$ is positively invariant for the semigroup, i.e., $S(t)\mathcal{M} \subseteq \mathcal{M}$, $\forall t\geq 0$. 
	\item $\mathcal{M}$ attracts all the orbits of (\ref{eq:dynsys1}) at exponential rate.
\end{enumerate}
An \emph{inertial system} for (\ref{eq:dynsys1}) is the system obtained by restricting the problem (\ref{eq:dynsys1}) to $\mathcal {M}$.
\end{definition}

In the case we are interested in, i.e.,  $ \mathcal{F}(u) = \Delta u - V'(u)$, we consider inertial manifolds in the form of a graph of a function $\Phi: \Pm\mathcal{H}\to\Qm\mathcal{H}$:
\begin{equation}
	\mathcal{M} = \graph(\Phi) = \set{u = \tonde{\mu,\eta} = \tonde{\mu,\Phi(\mu)} \in\mathcal{H}\colon \mu\in\Pm \mathcal{H}}.
\end{equation}

The meaning of condition $(i)$ is that every orbit starting in $\mathcal{M}$ must remain in $\mathcal{M}$. This can be rephrased by saying that equation (\ref{RD}) must be automatically satisfied along $\mathcal{M}$. Taking into account of the splitting (\ref{dec1}) of (\ref{RD}) along the eigenspaces of $-\Delta$, 
we obtain the following conditions on $\Phi$:
\begin{equation}        \label{eq:inertial_cond_system_phi}
	\begin{cases}
		\mu_t &= \Delta  \mu -\bP_m  V'(    \mu +\Phi (\mu)),  \\
		\Phi (\mu)_t &= \Delta \Phi (\mu) -\bQ_m  V'(    \mu +\Phi (\mu)), 
	\end{cases}
\end{equation}
which can be resumed in the following equation for $\Phi$:
\begin{equation}\label{eq:inertial_cond_uniq_phi}
	\Phi' (\mu)[\Delta \mu -\bP_m  V'(    \mu +\Phi (\mu))]=\Delta \Phi (\mu) -\bQ_m  V'(    \mu +\Phi (\mu)).
\end{equation}
We refer to equation (\ref{eq:inertial_cond_uniq_phi}) as the equation for the exact inertial manifold. 
As anticipated above, the solution of this equation is an ambitious goal, 
given that it 
leads to a finite ordinary differential equation (\ref{eq:inertial_cond_system_phi})$_{1}$ substituting the original PDE. 
It would be very desirable to have a general method for solving, at least numerically, equation (\ref{eq:inertial_cond_uniq_phi}).
However, this manifold is not guaranteed to exist in the general case, and even when the manifold exists, it may not be easily tractable by numerical methods. 

At this stage some heuristic considerations are helpful to get to the next step, consisting in deriving approximate versions of (\ref{eq:inertial_cond_uniq_phi}) and correspondingly to \emph{approximate inertial manifolds} (AIM). Because of the dissipative character of the dynamics, we have that for an arbitrary orbit $u(t):=(\mu(t),\eta(t))$,  the magnitude of the 
tail $\eta(t)$ becomes sensibly smaller than the magnitude of the head $\mu(t)$ after a transient period, i.e., $\abs{\eta(t)} \ll \abs{\mu(t)}$. 
This allows us to assume that $ 
	\abs{\Phi(\mu)} \ll \abs{\mu}$,
and to consider as a first trivial approximate inertial manifold the graph of the null function $\Phi(\mu)\equiv 0$ leading to the 
\begin{equation}
	\text{\emph{flat manifold}} := \set{(\mu,0)\colon \mu\in\Pm \mathcal{H}}\subseteq \mathcal{H}.
\end{equation}
 The corresponding inertial system would be then given by the finite equation
\begin{equation}
	\mu_t = \Delta \mu - \Pm V'(\mu). 
\end{equation}
In the same spirit, we  avoid to drop completely $\Phi(\mu)$ and all its derivatives from equation (\ref{eq:inertial_cond_system_phi}), but we rather  neglect only the terms $\Phi(\mu(t))$ and $\Phi'_t(\mu(t))$ in (\ref{eq:inertial_cond_system_phi})$_2$, obtaining the following equation
\begin{equation}\label{eq:phizero}
	0 =  \Delta \Phi (\mu) -\bQ_m  V'(    \mu ) \qquad \Longrightarrow \qquad 
	\Phi (\mu) := -(-\Delta)^{-1}\bQ_m  V'(    \mu ),
\end{equation}
which has the great advantage of giving an explicit formulation for $\Phi$, straightforwardly 
amenable to numerical computation. (See, e.g., (1.26) p. 569 in \cite{T}.)

The solution $\Phi_0$  of (\ref{eq:phizero}) is  called an \emph{approximate inertial manifold}. This is the primal instance of AIM, defined for the first time in the context of 2-dimensional Navier-Stokes equation \cite{T2} by C. Foias, O. Manley and R. Temam. 
This first appearance has been reused many times subsequently in both theoretical and numerical researches.

Two fundamental facts about $\Phi_0$  are worth mentioning  
\begin{enumerate}[(i)]
	\item Under some reasonable hypotheses, it is possible to prove the following fact \cite{T,Marion:1989}. Let $u(t) = (\mu(t),\eta(t))$ be any orbit of the full system. Then after a sufficiently long transient period the tail $\eta(t)$ will be smaller than the head $\mu(t)$, more precisely we have
	\[ d(u(t),\set{\text{\textit{flat manifold}}}) \cong \abs{\eta(t)} 
	\leq C \abs{\frac{\lambda_1}{\lambda_{m+1}}}.
	\]
	This means that the flat manifold $\set{(\mu,0)}$ is an approximation of order $ \abs{\frac{\lambda_1}{\lambda_{m+1}}}$. At the same time the distance of the orbit from the AIM $\mathcal{M}_0 = \textrm{graph}\Phi_0$ is estimated by
	\[ d(u(t),\mathcal{M}_0) \leq C \abs{\frac{\lambda_1}{\lambda_{m+1}}} ^2,
	\]
	which is sensibly smaller when $\lambda_{m+1}$ is large. Therefore the AIM $\mathcal{M}_0$ is a far better approximation of the exact inertial manifold if compared with the flat manifold. 
	\item It is possible to iterate the procedure giving rise to $\Phi_0$ for defining AIMs of higher order, $\Phi_1$, $\Phi_2$,\dots,$\Phi_k$,\dots, which can be proved to converge to the exact inertial manifold (see \cite[chap.\,X]{T}), as $k\to\infty$.
\end{enumerate}

\subsection{An alternative AIM from the static reduction.} 
\label{AIM2}
In view of $(\ref{dec2})_2$ it seems reasonable  to use as AIM the solution $\widetilde\eta$ of the equation related to the static problem
\begin{equation}\label{eq:psizero}
	0 =  \Delta \widetilde\eta (\mu) -\bQ_m  V'(    \mu + \widetilde\eta(\mu)).
\end{equation}
which is very close to the primal AIM $\Phi_0$ obtained in (\ref{eq:phizero}). 
Indeed, in Theorem \ref{teo:aimdist} below we want to prove that the manifold 
\begin{equation}
	\widetilde{\mathcal{M}}:= \mathrm{graph}(\widetilde\eta(\mu))=\set{(\mu,\widetilde\eta(\mu))\colon \mu\in \Pm \mathcal{H}}
\end{equation}
is an approximate inertial manifold with the same accuracy of $\Phi_0$. 
\begin{hypothesis}\label{hyp:vprime}
 	Assume that $V'$ is a Nemitsky operator, i.e., $V'(u) := \gamma(u)$, with $\gamma:\R\to\R$ compactly supported with Lipschitz constant $L_{V'}>0$. 
\end{hypothesis}

\begin{theorem}\label{teo:aimdist}
 	Assume that the hypothesis \ref{hyp:vprime} holds. Then, for $t$ sufficiently large, $t\geq t^\star$, any orbit of (\ref{RD}) remains at a distance in $\mathcal{H}$ of $\widetilde{\mathcal{M}}$ 
	bounded by 
	\[ d(u(t),\widetilde{\mathcal{M}}) \leq \kappa \abs{\frac{\lambda_1}{\lambda_{m+1}}} ^2,
	\]
	where $\kappa$ and $t^\star$ are appropriate constants depending on the data $\Omega, d, \gamma$ and on $R$ when $\abs{u_0} \leq R$. 
\end{theorem}
Preliminary considerations and a Lemma contained in  \cite{Marion:1989}.  We adapt to our case $\widetilde{\mathcal{M}}$ the argument used by \cite{Marion:1989} for the AIM $\mathcal{M}_0$. 
\label{AIMdistproof}
We begin by   recalling some notations and results introduced in \cite{Marion:1989}.
We set 
\begin{equation}
	\lambda := \lambda_m, \qquad \Lambda := \lambda_{m+1},\qquad \delta := \frac{\lambda_1}{\lambda_{m+1}}. 
\end{equation}

Clearly $(-\Delta)^{-1}$ is compact and we can define the powers $ (-\Delta)^\beta$ for any $\beta\in\R$. $D((-\Delta)^\beta)$ is a Hilbert space if endowed with the norm $\abs{(-\Delta)^\beta u}$, where $\abs{\cdot}$ denotes the norm of $L^2(\Omega)$. 

We have that $\Pm$ and $\Qm$ commute with $(-\Delta)^\beta$, for all $\beta\in\R$, and furthermore 
\begin{align}
 	\abs{(-\Delta)^{\beta+\frac{1}{2}} \mu}^2 \leq & \lambda \abs{(-\Delta)^{\beta} \mu}^2, & \forall \mu\in\Pm D((-\Delta)^\beta), \label{eq:condsVa}\\ 
 	\abs{(-\Delta)^{\beta+\frac{1}{2}} \eta}^2 \geq & \Lambda \abs{(-\Delta)^{\beta} \eta}^2, & \forall \eta\in\Qm D((-\Delta)^\beta). \label{eq:condsVb}
\end{align}

\begin{lemma}[See \cite{Marion:1989}]\label{lem:qqprimbounds}
 	Assume that hypothesis \ref{hyp:vprime} holds. Then there exists a time $t^\star$ and a constant $\kappa$ depending only on the data $(\Omega, d, V')$ and on $R$ when $\abs{u_0}\leq R$, such that  the tail $\eta(t) = \Qm u(t)$ of any orbit $u(t)$ of (\ref{RD}) becomes small in the following sense: 
	\begin{equation}
		\abs{\eta(t)} \leq \kappa\delta, \qquad \abs{\eta'(t)} \leq \kappa\delta, \qquad \forall t\geq t^\star.
	\end{equation}
\end{lemma}
\begin{proof}[Proof of Theorem \ref{teo:aimdist}]
	Let  $u= \mu+\eta$ be an orbit of (\ref{RD}). 
For every $t>0$, we set $\widetilde \eta(t) := \widetilde\eta(\mu(t))$.
We have that $\widetilde u(t) := \mu(t) + \widetilde \eta(t) \in \widetilde{\mathcal{M}}$, therefore, 
\begin{equation}
	\mathrm{dist}\tonde{u(t),\widetilde{\mathcal{M}}} \leq 
	\abs{\widetilde u(t) - u(t)} = \abs{\widetilde \eta(t) - \eta(t)}, 
\end{equation}
and for our purposes it suffices to estimate the norm in $\mathcal{H}$ of 
\begin{equation}
	\widetilde\chi(t) := \widetilde \eta(t) - \eta(t). 
\end{equation}
By definition of $\widetilde \eta$, (i.e.,  (\ref{dec2})$_2$) we have 
\begin{equation}
	-\Delta \widetilde \eta + \Q_m V' (\mu + \widetilde \eta) = 0, 
\end{equation}
then, subtracting (\ref{dec1})$_2$, we find
\begin{equation}
	-\Delta \widetilde\chi = \Qm V'(\mu+\eta) - \Qm V'(\mu+\widetilde \eta) + \eta', 
\end{equation}
from which, by Lemma \ref{lem:qqprimbounds}, it follows that 
\begin{equation}
	\abs{\Delta \widetilde\chi } \leq L_{V'}\abs{ \eta - \widetilde \eta} + \abs{\eta'}. 
\end{equation}
Since $\widetilde\chi\in\Qm\cH$, the estimate\footnote{By applying twice (\ref{eq:condsVb}), we write $\abs{\Delta^1 \chi}^2 \geq \Lambda \abs{\Delta^\frac{1}{2} \chi}^2 \geq 
\Lambda^2 \abs{\Delta^0 \chi}^2$, from which it follows $ \abs{\Delta \chi} \geq \Lambda \abs{\chi}$.
} 
(\ref{eq:condsVb}) holds and this leads to
\begin{gather}
	\Lambda \abs{\widetilde\chi} \leq   \abs{\Delta \widetilde\chi} \leq  L_{V'}\abs{ \widetilde \chi} + \abs{\eta'}\quad \Longrightarrow \\
	(\Lambda - L_{V'})\abs{\widetilde\chi} \leq  \abs{\eta'} \leq \kappa \delta \qquad \Longrightarrow \\
	\abs{\widetilde\chi} \leq  \frac{\kappa}{(\Lambda - L_{V'})} \delta = \kappa\frac{\Lambda}{\lambda_1(\Lambda - L_{V'})}\frac{\lambda_1}{\Lambda} \delta \leq \widetilde\kappa \delta^2
\end{gather}
which concludes the proof. 

\end{proof}

\goodbreak

\section{Fokker Planck equation and large deviations for the reduced gradient dynamics}

The finite dimensional gradient reduced equation $(\ref{gra})$ 
$$
\frac{d\mu}{dt}  = -{\nabla W}(\mu), \quad \mu \in \mathcal{D} \subseteq \mathbb{R}^m
$$
gives an approximate description of the dynamics:    in a broad sense, the reduced dynamics on the approximate inertial manifold stays close to the  true   system evolution, at least after a finite time interval and in the neighborhood of an equilibrium. Here $\mathcal{D}$ can be bounded, unbounded or the entire $\mathbb{R}^m$. Moreover in the following we will interpret the function $W$ introduced in $(\ref{Fr})$ as a sort of potential energy for the reduced system; secondly we will consider  a stochastic  version of the deterministic ODE $(\ref{gra})$, namely equation $(\ref{stogra})$ below. Implicitly we are assuming that the added noise will be able to mimic the loss of information introduced by considering the reduced dynamics in place of the original one\footnote{We recall that for the    reduced equation the set of critical points is  in a one-to-one correspondence with the equilibria of the original reaction--diffusion equation}. The closeness to the original system evolution of the stochastic trajectory will be dealt  with  in the framework of dynamic large deviation theory (Friedlin-Wentzell theory, see \cite{FW}). In this Section we give a cursory view of this approach.
Let us consider the SDE    associated to $(\ref{gra})$
\be
\label{stogra}
  \frac{d\mu}{dt} = - {\nabla W}(\mu) + \sqrt{\nu}\,w
\ee
where we   denote with $\nu>0$ the diffusion coefficient and  with $w$ the gaussian (white) noise;  its associated Fokker-Planck  equation for $p_\nu(t,\mu)$ in a spatial domain $\mathcal{D} \subseteq \mathbb{R}^m$   reads
 \be
\label{FPgra}
\frac{\partial p_\nu}{\partial t} - \nabla \cdot (p_\nu{\nabla W}) =  \nu \Delta p,  \ee
As it is well known, the stationary $(\frac{\partial p_\nu}{\partial t} \equiv 0)$ solution of the Fokker-Planck equation is
\be
\label{FPeq}
p_\nu^{eq} (\mu) = Z(\nu)^{-1}e^{-\frac{W(\mu)}{  \nu}},\qquad Z(\nu)= \int_\mathcal{D} e^{-\frac{W(\mu)}{  \nu}}d\mu
\ee
This solution exists  whenever $W$ is bounded from below and it grows rapidly enough to ensure that $Z$ is finite. 
The equilibrium solution is globally attractive and a Lyapunov function for it is given by the \emph{relative entropy}
\be\label{re}
 { H}(p_\nu \vert \, p^{eq}_\nu ) (t)= \int_\mathcal{D} p_\nu(t,\mu ) \ln  \frac{p_\nu(t,\mu)}{p^{eq}_\nu (\mu) }d\mu
\ee
 Moreover, it holds that in any finite time interval $[0,t]$, the trajectories of the SDE $(\ref{stogra})$ tend to the the solutions of the deterministic system $(\ref{gra})$ in the vanishing viscosity limit. This qualitative statement can be given a precise meaning using the language of Large Deviations. 
 \subsection{Large Deviations}
 
 We refer to classical literature on the subject for details and we state here only the main ideas. For sake of simplicity we  adopt hereafter the symbols $\mu = x$ and $X = -{\nabla W}$. 
 
 Let $\mathcal{D}^{[0,t]}$ be the set of all maps from the interval $[0,t]$ into $\mathcal{D}$; for a given initial condition $x(0) = x_0$, let us introduce the set 
 $$
  \mathcal{C}_\nu^{(x_0,t)} = \set{ x_\nu(\tau)\colon \tau \in [0,t], \ x(0) = x_0 } \subset \mathcal{D}^{[0,t]}
 $$
 of  continuos  trajectories  of the SDE $(\ref{stogra})$ that we rewrite as
 $$\frac{dx}{dt} = X(x) +   \sqrt{\nu}\,w$$ 
 and let us denote with $P_\nu$ the probability measure induced by the SDE on the measurable space $\mathcal{D}^{[0,t]}$.
Friedlin and Wentzell have shown that, if the vector field $X$ is Lipschitz continuous, the $\nu$-family of measures $P_\nu$ on $\mathcal{C}_\nu^{(x_0,t)}$ satisfies a large deviation principle in the form: for every $A \subset \mathcal{D}^{[0,t]}$
 it holds that
 $$
 \limsup_{\nu \rightarrow 0} \nu \ln P_\nu (A)  \leq - \inf_{x(\cdot ) \in \textit{cl}(A)}I_t[x(\cdot )], $$ 
 $$  \liminf_{\nu \rightarrow 0} \nu \ln P_\nu (A)  \geq - \inf_{x(\cdot )\in \textit{int}(A)}I_t[x(\cdot )]
 $$
  with respect to the Friedlin-Wentzell 
   \emph{action functional}   
\be
\label{actfunctional}
 I_t[x(\cdot )] = \frac{1}{2 }\int_0^t |\frac{dx}{dt}(\xi) - X(x(\xi)) |^2 d\xi,
\ee

In the sequel we will use the more compact notation  
 \be
 \label{LDP1}
  P_\nu [A] \asymp e^{-\frac{1}{\nu} \inf_{x(\cdot)\in A} I_t[x(\cdot)]}, 
 \ee
 where  
 $ I_t[x(\cdot)]$ is said \textit{rate function} and (after Ellis) the 
 symbol
 $\asymp$ stands for logarithmic equivalence
\be
a_\nu \asymp b_\nu \quad \Leftrightarrow \quad \lim_{\nu \rightarrow 0} \nu \ln a_\nu =  \lim_{\nu \rightarrow 0} \nu \ln b_\nu 
\ee
 
 The \textit{contraction principle} --a theorem by Varadhan, see e.g. \cite{Touchette}-- allows us to prove a large deviation principle for the probability density on $\mathcal{D}$, inherited from the above (\ref{LDP1}). In our case it does work in the following way. Given the map (projection)
 \be
 \pi :  \mathcal{C}_\nu^{(x_0,t)} \subset \mathcal{D}^{[0,t]}\longrightarrow \mathcal{D}\subseteq \mathbb{R}^m , \qquad \pi (x(\cdot)): = x(t),
 \ee
let us consider the \textit{push-forward} measure
 \be
 \label{prob}
 p_\nu\, dx = \pi_* P_\nu\, dA
 \ee
 The map $\pi$ is a random variable and the above introduced probability density $p_\nu$ is its  associated density. By the contraction principle, the $\nu$-family of probability densities $p_\nu$ satisfies a large deviation principle
  that  in compact notation reads
   \be
 \label{LDP2}
p_\nu^{x_0}(t,x) \asymp e^{-\frac{1}{\nu} \inf_{x(\cdot): \pi (x(\cdot))= x} I_t[x(\cdot)]}  \ee

with respect to the induced rate function
 \be
 \label{V}
 V(t, x, x_0) =   \mbox{inf}\   \set{I_t[x(\cdot )] \colon \  x(\cdot ) \in  \mathcal{C}_\nu^{(x_0,t)}, \ \pi (x(\cdot )) = x },
\ee
 such that we can   write
 \be
 \label{pgd2}
p_\nu^{x_0}(t,x) \asymp e^{-\frac{1}{\nu}V(t,x,x_0)};    \ee
 $V(t,x,x_0)$ is called the \emph{quasi-potential} and  
  \be
 \label{QP}
 L(x, \dot x) = \frac{1}{2 }| \dot x - X(x)|^2
\ee
 is the \emph{Lagrangian} of the stochastic process. 
 
 The above results open  the way to consider the solution of the deterministic system $(\ref{gra})$ between prescribed initial and final configurations as the optimal (in the sense of a least action principle) path. Following a standard reference \cite{BC}  on the subject, the starting point is the evolutive Hamilton-Jacobi equation
 \be
\label{HJ1}
 \frac{\partial S}{\partial t} + H\tonde{x, \nabla S} = 0
\ee
  for the Hamiltonian
\be
\label{hh}
 H(x,p) = \frac{1}{2}  |p|^2 + p\cdot X(x)
\ee
 associated to the above introduced Lagrangian (\ref{QP}). As it is well known,  the weak global solution of the  problem below
 \be
\label{HJ}
   \frac{\partial S}{\partial t}  + \frac{1}{2} |\nabla S|^2 + \nabla S\cdot X = 0, 
   \quad  S(0, x_0,x_0) =0   
   \ee
  can be written by the Lax-Oleinik viscosity solutions formula
\be
\label{S(t,x,x_0)}
 S(t,x,x_0) = \mbox{inf}  \set{ \int^t_0 L(x(\xi),\dot x(\xi)) d\xi \ \colon \ x(0) =x_0, \ x(t) = x }
\ee
Therefore we see immediately that \emph{the quasi-potential can be seen as the solution of a Hamilton-Jacobi equation}.

 \subsection{The gradient vector field case} We have seen before that if the vector field is a gradient,  $X = -{\nabla W}$, the equilibrium density $p^{eq}_\nu(x)$, up to normalization, is 
 $$
 p^{eq}_\nu(x) = e^{-\frac{ 1}{  \nu}W(x)}
 $$
In   the   Friedlin-Wentzell theory, in a neighborhood of an equilibrium point $\hat x$ , $X(\hat x) = 0$,  the equilibrium density $p^{eq}_\nu(x)$ has the form
 $$
 p^{eq}_\nu(x) \asymp e^{-\frac{1}{\nu}V^{eq}_\infty(x, \hat x)} 
 $$
with respect to  the quasi-potential 
$
 V^{eq}_\infty(x,\hat x) =\lim_{t \to +\infty} V(t, x, \hat x)
$,
\be
\label{veq}
 V^{eq}_\infty(x,\hat x)  = \inf \set{ \frac{1}{2 }\int_0^\infty |\dot x - X(x(\xi)) |^2 d\xi \ \colon \  x(\cdot ) \in \mathcal{C}_\nu^{(\hat x,\infty)}, \ \pi (x(\cdot ))=\lim_{t\to \infty}x(t) = x }
 \ee
We now show that for a gradient vector field $$X = - \nabla W$$ one has that $V^{eq}_\infty(x,\hat x) =  W(x)$ in a neighborhood of an equilibrium point. This can be seen   directly by expanding the Lagrangian
$$
| \dot x + {\nabla W}|^2 = \dot x^2 + 2 {\nabla W}\cdot \dot x + |{\nabla W}|^2 = \dot x^2 +   |{\nabla W}|^2 + 2\frac{dW}{dt},
$$
 hence, for every $t$,
 $$
 \frac{1}{2 }\int_0^t |\dot x (\xi)- X(x(\xi)) |^2 d\xi = \frac{1}{2 }\int_0^t (\dot x^2 + |{\nabla W}|^2)d\xi +   W(x) - W(\hat x) 
  $$
 It is not restrictive to suppose that $W(\hat x) = 0$. Moreover by the inf procedure after the limit $t \rightarrow +\infty$ we are left with
\be
\label{ris}
  V^{eq}_\infty(x,\hat x) =  W(x)   
\ee
 therefore, as expected, in the infinite time limit the large deviation description  of the probability density tends to the equilibrium solution of the F-P equation.
 \subsection{Viscosity solutions}

 It is instructive to point out another interesting relation between the above outlined large deviation theory and the viscosity solution approach to PDE. Inspired by the logaritmic equivalence $(\ref{pgd2})$ above  holding in the $\nu$-vanishing  limit,
 $$
 p_\nu^{x_0}(t,x) \asymp   e^{-\frac{1}{\nu}S(t,x,x_0)},
  $$
we  introduce the Cole-Hopf transformation (dropping $x_0$):
\be
\label{CH}
 p_\nu(t,x) = e^{-\frac{1}{\nu}\hat S_\nu(t,x)}
 \ee
By inserting this representation of $p_\nu(t,x)$ into the Fokker-Plank equation $(\ref{FPgra})$ above\footnote{Precisely, we consider the Fokker-Planck equation $(\ref{FPgra})$ with a factor $\nu/2$ in front of the Laplacian }, we see that the function $\hat S(t,x)$ must satisfy the equation
\be
\label{HJS}
 \frac{\partial \hat S_\nu}{\partial t}  + |\nabla \hat S_\nu|^2 + \nabla \hat S_\nu\cdot X =
  \nu \nabla \cdot X +  \frac{\nu}{2} \Delta \hat S_\nu
\ee
We can interpret it as a Hamilton-Jacobi equation with a perturbation $ \nu \nabla \cdot X$ in the Hamiltonian and  a \emph{viscosity} term $ \frac{\nu}{2}\Delta \hat S_\nu$. We stress the fact that the diffusion coefficient $\nu$ of the stochastic approach above is interpreted here as a viscosity term. Coupling the viscosity stability theorem (see \cite{BC}) and the   definition of viscosity solution,
we   know that, in the $\nu \rightarrow 0$   limit, the solution $\hat S_\nu$ of (\ref{HJS}) tends  in the ${\mathcal C}^0$ topology to the viscosity solution of (\ref{HJ}), which can be represented using the Lax-Oleinik formula
 $$
 \hat S(t,x) = \mbox{inf}   \set{ \frac{1}{2}\int_0^t |\dot x(\xi) - X(x(\xi)) |^2 d\xi \ \colon \ x(\cdot ) \in \mathcal{C}^1, \  x(0) = x_0, \ x(t) = x }
 $$ 
 Note that the value of the inf above does not depends on the degree of regularity we assume for the path  $x(\cdot)$.  Also,  standard facts from weak KAM theory, see e.g. \cite{D-S, F}, tell us that any viscosity solution of (\ref{HJ}), for $t\to +\infty$,  is ${\mathcal C}^0$-asymptotic (i.e.  in the uniform convergence topology)   to $\hat S(x)-ct$,  where
 $\hat{S}$ is a suitable viscosity solution of the related stationary H-J equation,
    \be
\label{SHJ}
 \frac{1}{2} |\nabla \hat S|^2 + \nabla \hat S \cdot X(x) = c
   \ee
 at the real value
$$c:=\inf_{u\in {\mathcal C}^1(\mathcal{D},\R)} \sup_{x\in\mathcal{D}} H(x, \nabla u(x))$$
called the \textit{Ma\~n\'e critical value}. 
 
 Now we recall that in our setting the Ma\~n\'e critical value is $c=0$; this fact is popular in the compact case (see below $(i)$), but it is also true in the non compact case, if $X$ is a gradient field admitting at least one critical value: this is precisely guaranteed by the assumption that $W$ is bounded from below, just as the functional $\mathcal{J}$ is.
\begin{proposition}\label{Prop1}\par\noindent$(i)$ Let $ \mathcal{D}$ be a $m$--dimensional closed Riemann manifold, and let $X:\mathcal{D}\to T\mathcal{D}, x\mapsto (x,X(x))$, be a vector field. Then the \textit{Ma\~n\'e critical value}  $c$ of the Hamiltonian $H=\frac{1}{2}\abs{p}^2 + p \cdot X(x)$ in (\ref{hh}) is vanishing:
\begin{equation}
	c=\inf_{u\in {\mathcal C}^1(\mathcal{D},\R)} \sup_{x\in\mathcal{D}} \left(\frac{\abs{\nabla u(x)}^2}{2} + \nabla u(x) \cdot X(x)\right) = 0
\end{equation}
\par\noindent$(ii)$ In the non compact case $\mathcal{D}=\R^m$ and for $X=-\nabla W$, if there exists some critical point $\hat x$, $\nabla W(\hat x)=0$, then the thesis $c=0$ is still true.
\end{proposition}
\begin{proof}
	Let $f_u(x)=\frac{1}{2}\abs{\nabla u(x)}^2 + \nabla u(x) \cdot X(x)$. First we see that
		 $$\inf_u\sup_x f_u(x) \leq 0,$$
		indeed, if $u\equiv 0$, $f_u(x)\equiv 0$ and therefore the $\inf_u$ must be equal or smaller than 0. Secondly,
		 $$\inf_u\sup_x f_u(x) \geq 0,$$
since $\mathcal{D}$ is a compact manifold, any $u$ must have at least one critical point $x_0$, $\nabla u(x_0) = 0$. Then $f_u(x_0) =0$ and $\sup_x f_u(x) \geq f_u(x_0) =0$. Because this holds for any $u$, therefore $\inf_u \sup_x f_u(x) \geq 0$. In order to prove $(ii)$, for any $u$, by evaluating at $x=\hat x$, we see that $f_u(\hat x)=\frac{1}{2}\abs{\nabla u(\hat x)}^2\geq 0$, and we conclude.
	
\end{proof}

  \section{Macroscopic field description of Nonequilibrium Thermodynamics} 
 \label{MFD}

 We recapitulate the theory developed so far: starting from a deterministic ODE $(\ref{gra})$, which gives the approximate evolution of our reaction-diffusion system, we introduced the associated SDE $(\ref{stogra})$ by adding    a  gaussian   noise term $\sqrt{\nu}w$. The set of stochastic (continuous) trajectories can be equipped with a probability density which satisfies a large deviation principle: the probability of a path is exponentially decreasing around the deterministic trajectory corresponding to vanishing  noise. The associated rate function $(\ref{actfunctional})$ is the \emph{action functional} of the Friedlin-Wentzell theory. An analogous large deviation results holds for the probability density $p_\nu(x,t)$ of the associated Fokker-Planck  equation. In this latter case, the associated rate function, called the \emph{quasi potential}, can be obtained as the solution of an Hamilton-Jacobi type PDE $(\ref{HJ})$.
 
 \subsection{Sketch of the theory of driven diffusive systems}

  In a series of papers  \cite{PRL, JL2009, IL, MFT}  G. Jona-Lasinio and coworkers expose a theory   of thermodynamic systems in a stationary non equilibrium state (driven diffusive systems) in terms of a finite number of space--time  macroscopic fields representing   thermodynamic variables. 
 It turns out that strictly analogous  objects  to  those above discussed, like the action functional,  the quasi potential and the Hamilton-Jacobi type equation do arise in the theory of driven diffusive systems developed by G. Jona-Lasinio et al., albeit in an infinite-dimensional setting. Moreover    the quasi-potential is interpreted as the non-equilibrium free energy of the system (see Sect.\,\ref{MFD} below).
Therefore, the reduced dynamic and its associated stochastic version presented in this paper may be interpreted as a sort of finite dimensional skeleton of the theory of driven diffusive systems.

In what follows we will sketch the results of this approach using the single scalar space time field of the macroscopic density $\rho (x,t)$, $x\in \mathcal{D}$. The theory has an axiomatic format that stems from a paradigmatic consideration of lattice gas models.

   \smallskip

 \textbf{Axiom} (\cite{JL2009}) \textit{The macroscopic evolution of the field $\rho(x,t)$ is given by the continuity equation
 \be
 \label{axiom1}
 \rho_t + \nabla \cdot j = 0
\ee
 together with the constitutive equation for the associated current $j $
\be
\label{axiom2}
 j (\rho) = -D(\rho) \nabla \rho + \chi(\rho) E
\ee
 where $D$ and $\chi$ are respectively the diffusion and mobility coefficients and $E(x,t)$ is the applied external field. }

 Appropriate boundary conditions are supplied and we are interested in the study of stationary non equilibrium solutions $\hat \rho$.

The key point is that, when considering the lattice gas model,  the above introduced continuity equation  is the hydrodynamic limit equation   for   density $\rho(x,t) $ which in turn is the limit of the empirical density $\rho_N  (z,t)$ of particles in a given point $z$ on the lattice $\Lambda \subset \mathbb{Z}^d$ with $N$ particles. Jona-Lasinio  and coworkers shows that  in the thermodynamic limit $N \rightarrow \infty$ the empirical density $\rho_N$ satisfies a Large Deviation principle. To introduce the associated rate function, we set
\be
\label{aaa}
I_t[\rho] = \frac{1}{4}\int_0^t dt \int_\mathcal{D} [\rho_t + \nabla \cdot j(\rho)] K(\rho)^{-1}[\rho_t + \nabla \cdot j(\rho)]dx \geq 0
\ee
  where the positive operator $K(\rho)$ is the analog of the diffusion coefficient $\nu$ in the
   finite-dimensional Friedlin-Wentzell theory and $K$ is defined by
  $$
  K(\rho) u = - \nabla \cdot (\chi(\rho)\nabla u), \qquad \forall  u \quad \mbox{such that}\  \ u= 0 \ \mbox{on} \ \partial \mathcal {D}.
  $$
 The functional $I_t$ represents the 'extra-cost' necessary to follow the trajectory $\rho$ and it is zero along the solutions of the continuity equation above.  Given a stationary solution $\hat \rho$ of the continuity equation, and another system evolution $\eta$ we set (we omit here the space dependence in the fields for simplicity's sake)
 \be
\label{enlibera}
 \mathcal{F} (\rho) = -\mbox{inf\,}_{\eta(\cdot)}\left\{ I_{\infty}[\eta]\ \colon \eta(0) = \hat \rho, \quad \lim_{t\to \infty}\eta(t) = \rho\right\}
\ee
 Jona-Lasinio and coworkers identify $\mathcal{F}(\rho) $ as the \emph{free energy of the system in the dynamic state} $\rho$. This last point is particularly interesting since a commonly accepted definition of free energy for a general system in a non equilibrium state is actually lacking in the literature.

 \subsection{Links between the gradient system and the Jona-Lasinio system}\label{confronto}
 \begin{itemize}
 
 \item Note that in the simplest case $D= const. $      the   continuity equation $(\ref{axiom1})$ above reads
 \be
 \label{axiom3}
 \rho_t  = - D\Delta \rho + \nabla \cdot (\chi E)
 \ee
 which is of the type considered in $(\ref{RD})$ although with a different source term.
 
 \item Also, the Lagrangian density function in $(\ref{aaa})$  is the equivalent of the Lagrangian  associated to the action functional in the finite dimensional (ODE) Friedlin-Wentzell theory and the quantity $\mathcal{F}(\rho) $ in $(\ref{enlibera})$ can   be seen as the infinite-dimensional analog of the quasi potential $V^{eq}(x)$ in $(\ref{veq})$. 
 
\item The above defined free energy function $\mathcal{F}$ parallels the definition of the quasi potential $V$, which is interpreted as the generating function of a canonical transformation in the Friedlin-Wentzell theory. Therefore  $\mathcal{F}$ can be seen as the solution of an infinite dimensional Hamilton-Jacobi equation of the form (see \cite{IL})
\be
\label{infinite}
\left\langle \nabla \frac{\delta \mathcal{F}}{\delta \rho}\cdot \chi(\rho) \nabla \frac{\delta \mathcal{F}}{\delta \rho} \right\rangle - 
\left\langle \frac{\delta \mathcal{F}}{\delta \rho} \nabla \cdot j(\rho) \right\rangle = 0
\ee
 where angular brackets stands for integration on the spatial domain $\mathcal{D}$. 
 \item   As a matter of fact,   in the finite dimensional case, see $(\ref{SHJ})$, using asymptotic theory of viscosity solutions we are able to prove that the right hand member of the finite-dimensional H-J equation is   the  Ma\~n\'e critical value  which  is zero in the present case. The need of setting the right hand member of $(\ref{infinite})$ equal to zero is not obvious in the Jona-Lasinio format.
 \end{itemize}
 \subsection{Free energy and stochastic dynamics}
 In this paragraph we want to draw attention to an interesting interpretation of the functional form of the non equilibrium free energy introduced in $(\ref{aaa})$. Let  $z$ be the phase space point  of a thermodynamical system and $E(z)$ the energy. It is well known that for a  discrete  thermodynamical system of average energy $e= \mathbb{E}_{\hat p}(E)$ in equilibrium with a thermal bath at temperature $T$  its statistical description is afforded by the Boltzmann-Gibbs distribution (here $k_B$ is Boltzmann constant)
 \be
\label{bg}
 \hat p(z) =  Z(T)^{-1}e ^{-\frac{E(z)}{k_BT} },  \quad Z(T) = \int e ^{-\frac{E(z)}{k_BT}}dz
\ee
 and its associated equilibrium free energy is  
\be
\label{eqfreen}
  \Psi (T) =  \mathbb{E}_{\hat p}(E)- k_BT H(\hat p) = - k_BT \ln Z 
\ee
 where $H(p) = -\int p_i \ln p_i$ is Shannon entropy.
 Note that the above definition of $\Psi$ is perfectly meaningful also for a different  probability density therefore we can define the  \emph{ non-equilibrium free energy} as
\be
\label{freen}
 \Psi (p) = \mathbb{E}_{ p}(E) - k_B T H(p)
\ee
 and it is easy to prove that  
 $$
\frac{ 1}{k_B T} [ \Psi (p) - \Psi (\hat p) ] =  \int p(z)\ln \frac{p(z)}{\hat p(z)} dz = {H} (p\vert\, \hat p) \geq 0
 $$
 where   ${H} (p\vert\, \hat p)$ is precisely the relative entropy (\ref{re}). 
  
   The relative entropy is therefore the quantity that measures the variation in free energy due to a fluctuation $p$ from the equilibrium $\hat p$. Since $H(p\vert\, \hat p)$ is a non negative quantity, we recover the well known fact that the free energy has a minimum in the stationary state with respect to finite deviations. 
   
   If we consider the system described by the SDE $(\ref{stogra})$ -- $(\ref{FPeq})$, we see that the equilibrium density $p_\nu^{eq}$ in $(\ref{FPeq})$ has the form of the Boltzmann-Gibbs distribution $(\ref{bg})$ if, at least formally, we   identify $i)$ the potential energy $W(x)$ of $(\ref{FPeq})$ with the system energy $E(z)$  and $ii)$ the diffusion  coefficient $\nu$ with the temperature $k_BT$.  Concerning $ii)$,  using the Einstein-Smoluchowsky relation 
 $$
 \nu = \mu\ k_BT 
 $$
 where $\mu $ is the (system dependent) mobility, we are allowed to interpret, up to a constant, the diffusion term $\nu$ as a 'stochastic temperature' and hence to \emph{define} the equilibrium free energy $(\ref{eqfreen})$ of the stochastic system     as   
 $$
 \hat \Psi (\nu) = \mathbb{E}_{\hat p}(W) -   \nu H(\hat p) =  -\nu \ln Z(\nu)
 $$
 and the \emph{out of equilibrium free energy} $(\ref{freen})$  as  
  $$
 \Psi_\nu (p)= \mathbb{E}_p(W)- \nu H(p).
 $$
 In the neighborhood of an equilibrium point $\hat x$ of $W$, we have -see $(\ref{veq})$ and $(\ref{ris})$- that $W  = V^{eq}_\infty  $ hence the out of equilibrium free energy of our stochastic system is 
\be
\label{enlibera2}
 \Psi_\nu (p)= \mathbb{E}_p(V^{eq}_\infty)- \nu H(p)
\ee
   to be compared with $(\ref{enlibera})$. It is worth noticing, however, that this proposed free energy functional is, in this form,  a function of the probability density and not of the macroscopic parameters (pressure, volume, etc.) that define the thermodynamic state. This last step can be carried over by e.g. applying the maximum entropy principle where  the probability density $p_{eq}$ is a function of the observed macroscopic parameters. In this way we have a complete correspondence, at least in principle,  between 
   $(\ref{enlibera})$ and $(\ref{enlibera2})$.

 
\def\cprime{$'$} \def\cprime{$'$}

\end{document}